\documentclass[letterpaper, 10 pt, conference]{ieeeconf}  %

\IEEEoverridecommandlockouts                              %
\usepackage{amsmath} %
\usepackage{amssymb}  %
\usepackage{algorithmic}
\usepackage{algorithm}
\usepackage{pgfplots}
\usepackage{booktabs}
\pgfplotsset{compat=newest}
\usepackage{multirow}
\usetikzlibrary{plotmarks}
\usetikzlibrary{arrows.meta}
\usepgfplotslibrary{patchplots}
\usepackage{grffile}
\pgfplotsset{plot coordinates/math parser=false}
\newlength\figureheight
\newlength\figurewidth
\newcommand{\isMainDocument}{}
\usepackage{tikz}
\usepackage{pgfplots}
\usepackage{graphicx}
\usepackage{standalone} %
\usepackage{siunitx}
\pgfplotsset{compat=1.5}
\usepgfplotslibrary{units}
\usepgfplotslibrary{groupplots}
\usetikzlibrary{positioning,intersections, hobby, patterns, calc, decorations.pathmorphing, decorations.markings, shadows, shapes, plotmarks, shapes.multipart}
\usetikzlibrary{fit}
\usetikzlibrary{arrows.meta}
\usepackage{tikz}
\usepackage{amsmath}
\usepackage{bm}
\usepackage{threeparttable}
\usepackage{svg}
\usepackage[T1]{fontenc}
\usetikzlibrary{shapes.geometric, arrows.meta, positioning, fit, backgrounds, decorations.pathreplacing, calc}

\usepgfplotslibrary{external}

\title{\LARGE \bf
Bridging RL and MPC for mixed-integer optimal control with application to Formula 1 race strategies
}

\author{{Joschua W\"{u}thrich, Romir Damle, Giona Fieni, Melanie N. Zeilinger, Christopher H. Onder, Andrea Carron}%
\thanks{All authors are with the Institute for Dynamic Systems and Control, Department of Mechanical and Process Engineering, ETH Zurich, Switzerland \tt\small \{jwuethrich, rdamle, gfieni, mzeilinger, onder, carrona\}@ethz.ch}}

\begin{document}

\pgfplotsset{
            /pgfplots/layers/niceLayers/.define layer set={
                        axis background,
                        axis grid,
                        main,
                        axis ticks,
                        axis lines,
                        axis tick labels,
                        axis descriptions,
                        axis foreground
            }{/pgfplots/layers/standard}
}

\pgfplotsset{
            every axis/.append style={
                        set layers=niceLayers,
                        tick label style={font=\scriptsize},
                        clip marker paths=true,
                        line width=1pt,
                        line cap=round,
                        line join=round,
                        tick style={semithick, color=black},
                        legend style={
                        		   font=\scriptsize,
                                    /tikz/every even column/.append style={column sep=2mm},
                                    cells={anchor=west}, %
                        },
                        xmajorgrids,
                        ymajorgrids,
            }
}

\pgfkeys{/pgf/number format/.cd,1000 sep={}}

\maketitle
\thispagestyle{empty}
\pagestyle{empty}

\renewcommand{\thefootnote}{}
\footnotetext{\hspace{-1.5em} © 2026 IEEE. Personal use of this material is permitted. 
Permission from IEEE must be obtained for all other uses, in any current or future media, 
including reprinting/republishing this material for advertising or promotional purposes, 
creating new collective works, for resale or redistribution to servers or lists, or reuse 
of any copyrighted component of this work in other works.}
\renewcommand{\thefootnote}{\arabic{footnote}}
\begin{abstract}
We propose a hybrid reinforcement learning (RL) and model predictive control (MPC) framework for mixed-integer optimal control, where discrete variables enter the cost and dynamics but not the constraints. Existing hierarchical approaches learn a policy only for the discrete action space, leaving continuous optimization to MPC. Unlike these methods, we train an RL agent on the \textit{full} hybrid action space, such that the learned critic approximates the Q-function of the underlying Markov decision process. During deployment, the RL actor is rolled out over the prediction horizon to parametrize an integer-free nonlinear MPC through the discrete action sequence and provide a continuous warm-start. The learned critic can additionally serve as a terminal cost to capture long-term performance. We prove recursive feasibility, and validate the framework on a Formula 1 race strategy problem, where an ablation study identifies the contribution of each learned component. The hybrid method achieves near-optimal performance relative to an offline mixed-integer nonlinear program benchmark, outperforming a standalone RL agent. Moreover, the hybrid scheme enables adaptation to unseen disturbances through modular MPC extensions at zero retraining cost.
\end{abstract}

\section{INTRODUCTION}
Many real-world systems must simultaneously make discrete and continuous decisions. Heating and cooling systems combine discrete operating mode selection, such as direct solar supply, thermal storage discharge, and natural ventilation, with continuous control of heat pump power \cite{HybridHVAC}. Industrial process control handles discrete decisions such as valve switching, pump scheduling, or reactor mode selection and continuous regulation of flow rates, temperatures and pressures \cite{NPComplete}. Decision-making problems such as race strategies in Formula 1 (F1) racing \cite{F1RaceStrategy} optimize the continuous energy management of the hybrid power unit together with discrete pit stop choices.

From a control point of view, such systems naturally give rise to mixed-integer optimal control problems (MIOCPs). Model predictive control (MPC) \cite{MPC} and reinforcement learning (RL)~\cite{RL} are two distinct approaches for solving optimal control problems (OCPs) with complementary strengths. MPC provides safety and feasibility guarantees through constrained optimization problems, while RL naturally handles long-horizon sequential decision making.

In this paper, we propose a hybrid RL-MPC framework that exploits the complementary strengths of both methods for mixed-integer optimal control.

\subsection{Related work}
Recently, many research efforts have been made to combine RL and MPC in a unified framework for general OCPs to exploit their complementary strengths. RL commonly learns an approximation of the value function through data-driven exploration, capturing long-horizon behavior across the full state space, while MPC refines this approximation online through constrained optimization, replanning at each time step from the current state \cite{BertsekasAlphaZero, Mesbah}. This global-local perspective has emerged as a unifying lens for understanding RL-MPC integration, where RL supplies long-term information and MPC refines decisions locally under constraints. Three main streams of possible RL-MPC synthesis approaches can be distinguished: value-augmented MPC, differentiable MPC, and hierarchical RL-MPC schemes \cite{ReiterSurvey}.

Value-augmented MPC integrates a learned value function as a terminal cost in the MPC objective, allowing the reduction of the required planning horizon while improving long-term performance \cite{DeepValueMPC, PredictiveControlValueIteration, AC4MPC}. However, existing formulations do not exploit mixed-integer structure, and the computational cost of mixed-integer optimization remains.

Differentiable MPC replaces neural network function approximators with constrained optimization layers within policy-gradient or actor-critic updates \cite{Gros_Zanon_ENMPC, DiffableMPC_Amos}, yielding more interpretable policies with explicit constraint handling. However, constrained optimization problems must be solved during both offline training and online deployment, which is computationally demanding, and formulating MIOCPs in this setting \cite{Gros_Zanon_MIOCP} increases the burden further.

Hierarchical RL-MPC schemes decompose the problem across two levels, with RL approximating a global policy and MPC serving as a local refinement module. Examples include guided policy search \cite{GuidedPolicySearch}, MPC-Net \cite{MPC-Net}, and related frameworks in which an MPC generates trajectories used for policy learning \cite{ImitationLearning}. However, these methods do not specifically address the hybrid action space structure of mixed-integer OCPs. 

Existing hierarchical schemes that address the structure of hybrid action spaces learn a policy that maps the current state to a discrete action trajectory over the prediction horizon~\cite{daSilva, Dong}. This trajectory parametrizes the MPC, which optimizes only the continuous actions. In the underlying hybrid Markov decision process (MDP), the Q-function maps a state and a single hybrid action to the expected cumulative future reward, also called value~\cite{RL}. Neither approach learns such a Q-function. In~\cite{daSilva}, the learned Q-functions map the current state and a predicted discrete action trajectory to a value, and are decoupled across prediction horizon steps. They are thus defined on a different domain than the Q-function of the hybrid MDP, and the discrete decisions are not optimized against the latter. Moreover, since their input dimension is tied to the prediction horizon, changing the horizon requires retraining. In~\cite{Dong}, the policy is obtained by supervised imitation of a mixed-integer MPC, so no approximation of the Q-function is learned and the policy's performance is tied to the demonstrations it is fitted~to.

\subsection{Contribution} 
The contributions of this paper are threefold.

First, we propose a hybrid RL-MPC framework that trains the RL agent on the full hybrid action space, such that the learned critic approximates the Q-function of the underlying hybrid MDP. We focus on systems where discrete inputs affect the cost and dynamics but not the feasibility of constraints. During deployment, the RL actor is rolled out over the MPC prediction horizon using the system model to parametrize and warm-start a continuous MPC. The RL critic can additionally serve as a terminal cost, capturing long-term performance beyond the planning horizon. This reduces the online problem to a continuous OCP, enabling efficient deployment.

Second, we prove that under a structural assumption on the integer decisions, the proposed framework guarantees closed-loop constraint satisfaction. 

Third, we validate the framework on an F1 race strategy case study, evaluating performance and adaptability of the hybrid RL-MPC framework, including an ablation study of the learned components.

\subsection{Outline}
The remainder of this paper is organized as follows: Section \ref{sec:preliminaries} introduces the class of problems addressed by our method, their MDP formulation, and basics on RL and MPC. The proposed framework and its closed-loop constraint satisfaction property are presented in Section \ref{sec:framework}. We validate our method with an F1 race strategy case study in Section~\ref{sec:casestudy} and conclude our work in Section \ref{sec:conclusion}.

\section{PRELIMINARIES}
\label{sec:preliminaries}
We consider dynamic systems with hybrid action spaces that give rise to MIOCPs. This section formalizes the problem class that can be addressed with our proposed framework, and briefly introduces RL and MPC as complementary solution methods. Throughout this paper, $j\in \mathbb{N}$ denotes the closed-loop time step and $k\in\{0,...,N\}$ the open-loop prediction index of a finite-horizon problem. Closed-loop states, continuous and discrete inputs are denoted by $x(j), u^\mathrm{c}(j), u^\mathrm{d}(j)$, respectively, while open-loop predicted quantities are denoted by ${x}_{k}, {u}_{k}^\mathrm{c}, {u}_{k}^\mathrm{d}$. For the remainder of this paper, we will follow the control convention and work with cost functions rather than value functions. For a given policy $\pi$, the cost function~$J^\pi(x)$ is the negative of the value function in standard RL literature, and the state-action cost function $C^\pi(x,u)$ is the negative of the Q-function.

\subsection{Problem formulation}
We consider deterministic, discrete-time systems of the form
\begin{equation}
    x(j+1) = f \left(x(j), u^\mathrm{c}(j), u^\mathrm{d}(j) \right),
\end{equation}
where $x(j) \in \mathcal{X} \subset \mathbb{R}^{n_\mathrm{x}}$ is the state, $u^\mathrm{c}(j) \in \mathcal{U}_\mathrm{c} \subset \mathbb{R}^{n_\mathrm{c}}$ is the continuous input, $u^\mathrm{d}(j) \in \mathcal{U}_\mathrm{d} \subset \mathbb{Z}^{n_\mathrm{d}}$ is the discrete input, with $\mathcal{U}_\mathrm{d}$ finite, and $f: \mathcal{X}\times \mathcal{U}_\mathrm{c}\times \mathcal{U}_\mathrm{d} \rightarrow \mathcal{X}$ denotes the system dynamics.

For these systems, we seek policies $\pi(u|x)$, which are in general stochastic and minimize the discounted infinite-horizon cost 
\begin{equation}
    J^\pi\left(x(0)\right)=\mathbb{E}_\mathrm{\pi} \left[ \sum_{j=0}^{\infty} \gamma^j \ell \left(x(j), u^\mathrm{c}(j), u^\mathrm{d}(j)\right) \right],
    \label{eq:MDPcost}
\end{equation}
with stage cost $\ell: \mathcal{X} \times \mathcal{U}_\mathrm{c} \times \mathcal{U}_\mathrm{d} \rightarrow \mathbb{R}$ and discount factor~$\gamma \in~(0,1]$, where episodic problems of finite length are covered by an absorbing terminal state with zero stage cost.

The specific problem class that can be addressed with our framework is defined through the following assumption.
\newtheorem{assumption}{Assumption}
\begin{assumption}
    For any $x(j) \in \mathcal{X}$ and $u^\mathrm{d}(j) \in \mathcal{U}_\mathrm{d}$, there exists $u^\mathrm{c}(j) \in \mathcal{U}_\mathrm{c}$ such that $f\left(x(j), u^\mathrm{c}(j), u^\mathrm{d}(j)\right) \in \mathcal{X}$.
    \label{ass:one}
\end{assumption}
\newtheorem{remark}{Remark}
\begin{remark}
    Assumption \ref{ass:one} ensures that feasibility of the optimization problem can always be maintained through the continuous control input, regardless of the discrete input chosen by the RL agent. This is necessary as standard RL methods provide no inherent constraint satisfaction guarantees \cite{SafeRL}.
    This assumption is satisfied in many engineering applications, particularly in economic cost settings where discrete inputs select operating modes that affect system performance but not the feasibility of physical constraints. However, it can be restrictive when some discrete inputs directly limit the reachability of $\mathcal{X}$, e.g., by disabling an actuator needed for constraint satisfaction.
\end{remark}

The combinatorial structure of this problem class makes online mixed-integer optimization challenging.

\subsection{Markov decision process reformulation}
\label{subsec:MDP}
The presented MIOCP can be reformulated as a hybrid MDP. Denoting the joint action $u=\left(u^\mathrm{c}, u^\mathrm{d}\right)\in\mathcal{U}_\mathrm{c}\times \mathcal{U}_\mathrm{d}$, the MDP is defined by the previously introduced definitions for the state and action space, state transition function, stage cost, and discount factor
\begin{equation}
    \mathcal{M} = \left(\mathcal{X}, \mathcal{U}_\mathrm{c} \times \mathcal{U}_\mathrm{d}, f, \ell, \gamma \right),
    \label{eq:MDP}
\end{equation}
where the Markov property follows directly from the system dynamics. For a given policy $\pi(u|x)$, the state-action cost function of the MDP is
\begin{equation}
     C^\pi(x,u)=\mathbb{E}_\mathrm{\pi} \left[ \sum_{j=0}^{\infty} \gamma^j \ell\left(x(j), u(j)\right) \bigg| \, \begin{aligned} x(0)&=x,\\ u(0)&=u \end{aligned} \right],
     \label{eq:critic}
\end{equation}
from which the cost function \eqref{eq:MDPcost} follows as
\begin{equation}
    J^\pi(x) = \mathbb{E}_\mathrm{u\sim\pi(\cdot|x)} \left[ C^\pi(x,u)\right].
    \label{eq:cost}
\end{equation}
The optimal cost function $J^*(x)$ satisfies the Bellman equation 
\begin{equation}
    J^*(x) = \min_{u\in\mathcal{U}_\mathrm{c}\times\mathcal{U}_\mathrm{d}} \left[\ell(x,u) + \gamma J^*\left(f(x,u)\right) \right],
    \label{eq:Bellman}
\end{equation}
which also holds for the optimal state-action cost function~$C^*(x,u)$
\begin{equation}
    C^*(x,u) = \ell(x,u) + \gamma \min_{u'\in \mathcal{U}_\mathrm{c}\times \mathcal{U}_\mathrm{d}} C^*\left(f(x,u), u'\right).
    \label{eq:bellman_optimality_eqn}
\end{equation}

Solving the Bellman equation, \eqref{eq:Bellman} or \eqref{eq:bellman_optimality_eqn}, through exact dynamic programming (DP) becomes intractable due to the curse of dimensionality arising from the continuous state and action spaces \cite{Bertsekas}. In the subsections below, we present two methods that solve approximations to these equations.

\subsection{RL for mixed-integer control}
Depending on the algorithm, RL approximately solves either \eqref{eq:Bellman} or \eqref{eq:bellman_optimality_eqn} through data-driven exploration of the state and action spaces. $C^*(x,u)$ and $J^*(x)$ are approximated by parametrized function approximators $C_\mathrm{\theta}(x,u)$ and $J_\mathrm{\theta}(x)$, which are typically represented as neural networks. In the following, we present all expressions in terms of $C_\mathrm{\theta}(x,u)$, consistent with the RL algorithm used in our numerical study in Section~\ref{sec:casestudy}. However, all statements hold analogously for $J_\mathrm{\theta}(x)$. The parameters $\theta$ of the function approximator~$C_\mathrm{\theta}(x,u)$ are updated by minimizing the expected squared Bellman residual
\begin{equation}
\begin{split}
    \min_{\theta} \; \mathbb{E}_{(x,u)\sim \mathcal{D}} \big[\big(&C_\mathrm{\theta}(x,u)-\ell(x,u) \\
    &-\gamma \min_{u'}C_{\Bar{\theta}}\big(f(x,u),u'\big)\big)^{2}\big],
\end{split}
    \label{eq:rl_bellman}
\end{equation}
where $\mathcal{D}$ denotes the distribution of state-action pairs collected through the MDP $\mathcal{M}$ \eqref{eq:MDP} and stored in a replay buffer, and $\Bar{\theta}$ are target parameters held fixed during the update. Since the dynamics $f$ are deterministic, the expectation is taken over $\mathcal{D}$ alone. Given an approximation $C_\mathrm{\theta}(x,u)$, a parametrized policy $\pi_\phi(u|x)$ can be obtained through
\begin{equation}
    \phi^* = \arg  \min_{\phi} \; \mathbb{E}_\mathrm{x\sim\mathcal{D}, \, u\sim\pi_\phi(\cdot|x)} \left[C_\mathrm{\theta}\left(x,u\right)\right].
\end{equation}

Connecting back to the global-local framework \cite{BertsekasAlphaZero, Mesbah}, RL globally approximates the MDP cost \eqref{eq:MDPcost} through \eqref{eq:critic} and~\eqref{eq:cost} over the state space by training on the full hybrid action space $\mathcal{U}_\mathrm{c} \times \mathcal{U}_\mathrm{d}$. Optimization of the parameters $\theta$ relies on exploration to capture long-term effects of continuous and discrete actions, and minimizing \eqref{eq:rl_bellman} approximates the state-action cost function of the underlying hybrid MDP.

In actor-critic RL algorithms \cite{ActorCritic}, the policy $\pi_\mathrm{\phi}(u|x)$ serves as the actor, selecting actions, while $C_\mathrm{\theta}(x,u)$ serves as the critic, evaluating them and guiding the actor's updates during training. 

Two properties of RL are particularly relevant for mixed-integer control. Rather than explicitly enumerating and directly optimizing over the discrete action space as in branch-and-bound methods \cite{branch_and_bound}, RL optimizes the parameters $\theta$ and $\phi$ of $C_\mathrm{\theta}(x,u)$ and $\pi_\mathrm{\phi}(u|x)$, from which discrete actions follow implicitly. Furthermore, the dominant computational effort is shifted to offline training. Online deployment is reduced to a single forward pass through $\pi_\mathrm{\phi}(u|x)$, independent of the length of the optimization horizon.

However, standard RL alone does not enforce state or input constraints, and approximation errors in $C_\mathrm{\theta}(x,u)$ or $\pi_\mathrm{\phi}(u|x)$ may lead to suboptimal performance, where function approximators lack the resolution of model-based optimization.

\subsection{MPC for mixed-integer control}
MPC approximates the MDP cost \eqref{eq:MDPcost} locally by repeatedly solving the finite-horizon MIOCP
\begin{equation}
    \begin{aligned}
        \min_{ \{{x}_{k}, {u}_{k}^\mathrm{c}, {u}_{k}^\mathrm{d} \}} & \sum_{k=0}^{N-1} \ell\left({x}_{k}, {u}_{k}^\mathrm{c}, {u}_{k}^\mathrm{d}\right) + J_\mathrm{f}\left({x}_{N}\right) \\
        \text{subject to} \quad & {x}_{k+1} = f\left({x}_{k}, {u}_{k}^\mathrm{c}, {u}_{k}^\mathrm{d}\right) \\
        & {x}_0 = x(j) \\
        & {x}_{k} \in \mathcal{X}, x_{N} \in \mathcal{X}_\mathrm{f} \\
        & {u}_{k}^\mathrm{c} \in \mathcal{U}_\mathrm{c}, {u}_{k}^\mathrm{d} \in \mathcal{U}_\mathrm{d},
    \end{aligned}
    \label{eq:MPC_OCP}
\end{equation}
which is initialized at the current closed-loop state~$x(j)$, where $J_\mathrm{f}(x)$ is a terminal cost approximating the cost-to-go beyond the prediction horizon $N$, and $\mathcal{X}_\mathrm{f} \subseteq \mathcal{X}$ is a terminal set ensuring recursive feasibility. The horizon length~$N$ trades off computational cost with approximation quality. At each closed-loop time step $j$, only the first optimal action~$({u}_{0}^\mathrm{c,*}, {u}_{0}^\mathrm{d,*})$ is applied to the system before replanning over a fixed horizon~$N$, yielding a feedback policy through receding-horizon control. 

MPC offers strong theoretical properties in constrained settings, including recursive feasibility and closed-loop constraint satisfaction under standard assumptions \cite{MPC}. However, in a general formulation with a hybrid action space, the OCP \eqref{eq:MPC_OCP} is a mixed-integer nonlinear program (MINLP), which is NP-hard \cite{NPHard}. Increasing horizon length improves approximation quality but significantly increases computational complexity.

\section{HYBRID RL-MPC FRAMEWORK}
\label{sec:framework}
The complementary strengths identified in Section \ref{sec:preliminaries} motivate a structured combination of RL and MPC for mixed-integer control. RL provides a global cost approximation, while MPC offers model-based optimization capable of enforcing constraints and refining continuous actions locally. Under Assumption \ref{ass:one}, discrete decisions influence performance but not feasibility. We first present the policy architecture before establishing recursive feasibility of the closed-loop scheme.

\subsection{Policy architecture}
\label{subsec:policy}
In our proposed framework, an RL agent is trained offline via an actor-critic algorithm on the hybrid action space~$\mathcal{U}_\mathrm{c}\times~\mathcal{U}_\mathrm{d}$ of the underlying MDP. This yields an actor $\pi_\mathrm{\phi}(u|x)$ and a critic $C_\mathrm{\theta}(x,u)$. Following the control convention introduced in Section \ref{sec:preliminaries}, $C_\mathrm{\theta}(x,u)$ corresponds to the negative of the Q-function of standard RL literature. Since the agent is trained on the full hybrid action space, the resulting critic approximates the state-action cost function of the underlying hybrid MDP, contrary to methods that learn trajectory-conditioned or decoupled Q-functions \cite{daSilva}. Therefore, the actor $\pi_\mathrm{\phi}(u|x)$ is also informed by the hybrid nature of the problem, ensuring that the policy reflects the joint structure of the hybrid action space~$\mathcal{U}_\mathrm{c}\times \mathcal{U}_\mathrm{d}$.

With the actor and critic obtained offline, we present the closed-loop online deployment in Fig. \ref{fig:rl_mpc_framework}. At each closed-loop time step $j$, the RL actor $\pi_\phi(u|x)$ is rolled out deterministically over the prediction horizon $N$ using the system model~$f$, starting from the current state $x(j)$. This rollout yields a discrete input sequence  
\begin{equation}
    \left\{{u}_{k}^\mathrm{d,RL}\right\}_{k=0}^N,
    \label{eq:disc_rollout}
\end{equation}
and a corresponding continuous sequence $\{{u}_{k}^\mathrm{c,RL}\}_{k=0}^N$ that we can use as a warm-start, allowing the MPC to refine it. Given the discrete input sequence \eqref{eq:disc_rollout}, the MPC formulation~\eqref{eq:MPC_OCP} is parametrized and reduced to the following integer-free nonlinear program (NLP) formulation
    \begin{align}
        \min_{ \{{x}_{k}, {u}_{k}^\mathrm{c} \}} & \sum_{k=0}^{N-1} \ell\left({x}_{k}, {u}_{k}^\mathrm{c} \,| \,{u}_{k}^\mathrm{d,RL}\right) + C_\mathrm{\theta}\left({x}_{N}, {u}_{N}^\mathrm{c} \, |\, {u}_{N}^\mathrm{d,RL}\right) \nonumber \\
        \text{subject to} \: & {x}_{k+1} = f\left({x}_{k}, {u}_{k}^\mathrm{c} \, | \,{u}_{k}^\mathrm{d,RL}\right) \nonumber\\
        & {x}_0 = x(j) \label{eq:reduced_MPC}\\
        & {x}_{k} \in \mathcal{X}, x_{N} \in \mathcal{X}_\mathrm{f}\nonumber \\
        & {u}_{k}^\mathrm{c} \in \mathcal{U}_\mathrm{c}, \nonumber
    \end{align}
where the terminal cost can be given by the learned critic,~$C_\mathrm{\theta}({x}_{N}, {u}_{N}^\mathrm{c}\,|\, {u}_{N}^\mathrm{d,RL})$, evaluated at the terminal action $({u}_N^\mathrm{c}, {u}_N^\mathrm{d,RL})$. This extends the input sequence to $k=N$ rather than $N-1$, making ${u}_N^\mathrm{c}$ an additional decision variable compared to standard receding-horizon formulations, which is not required for a terminal cost of the form $J_\mathrm{\theta}(x_N)$. Whether the critic is beneficial as a terminal cost is problem-dependent~\cite{AC4MPC}, and examined empirically in Section~\ref{sec:casestudy}.

\begin{figure}
  \centering
  \vspace{0.2cm}
  \resizebox{0.95\columnwidth}{!}{%
    \input{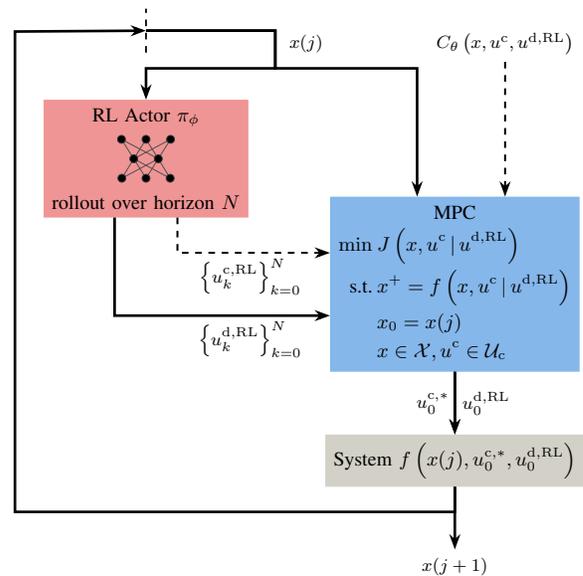}
  }
  \vspace{-0.2cm}
  \caption{Hybrid RL-MPC framework. The RL actor is rolled out over the prediction horizon $N$ to supply discrete action trajectories and a continuous warm-start to the MPC, which optimizes continuous inputs with the critic~$C_\theta(x,u^\mathrm{c},u^\mathrm{d})$ as a potential terminal cost. Only the first continuous input ${u}_{0}^\mathrm{c,*}$ from the MPC and the first discrete input ${u}_{0}^\mathrm{d,RL}$ from the RL agent are applied to the system.}
  \label{fig:rl_mpc_framework}
  \vspace{-0.5cm}
\end{figure}

Once the MPC is solved, only the first discrete input ${u}_{0}^\mathrm{d,RL}$ from the RL rollout and the first continuous input ${u}_{0}^\mathrm{c,*}$ from the MPC solution are applied to the system in a receding-horizon fashion.

\begin{remark}
    The continuous warm-start sequence provided by the RL actor initializes the NLP solver, such as interior-point \cite{InteriorPointMethods} or sequential quadratic programming (SQP) methods \cite{SQP}. This warm-start reduces solver iterations and may lower the risk of convergence to poor local minima, while allowing the MPC layer to refine the sequence through model-based optimization.
\end{remark}

\subsection{Recursive feasibility}
Recursive feasibility guarantees that once the hybrid scheme is initialized in a feasible state, the optimization problem remains feasible at all subsequent closed-loop time steps. In the following, we establish recursive feasibility for the receding-horizon implementation of the hybrid RL-MPC scheme, which relies on the slightly modified standard MPC assumption \cite{MPC} specified next.

\begin{assumption}
    There exists a terminal set $\mathcal{X}_\mathrm{f} \subseteq \mathcal{X}$ and a terminal control law $\kappa_\mathrm{f}(x)$ such that $\mathcal{X}_\mathrm{f}$ is positively invariant under $\kappa_\mathrm{f}$ for any $u^\mathrm{d} \in \mathcal{U}_\mathrm{d}$, all state and input constraints are satisfied in $\mathcal{X}_\mathrm{f}$, and the $N$-step backward reachable set of $\mathcal{X}_\mathrm{f}$ is independent of $u^\mathrm{d}$.
    \label{ass:two}
\end{assumption}

\newtheorem{proposition}{Proposition}
\begin{proposition}
    Let Assumptions \ref{ass:one} and \ref{ass:two} hold. If the reduced MPC problem \eqref{eq:reduced_MPC} is feasible at time $j=0$, then it remains feasible for all $j \in \mathbb{N}$, and consequently the hybrid RL-MPC scheme is recursively feasible.
\end{proposition}

\begin{proof}
    By Assumptions \ref{ass:one} and \ref{ass:two}, feasibility is independent of the discrete input sequence generated by the RL actor. The proof therefore follows the standard shifting argument for receding-horizon MPC under Assumption \ref{ass:two} \cite{MPC}, applied solely to the continuous component. Note that while our MPC formulation \eqref{eq:reduced_MPC} optimizes over inputs up to $k=N$ due to the terminal cost, the standard argument for recursive feasibility requires that $x_{N} \in \mathcal{X}_\mathrm{f}$, which depends only on the input sequence up to $k=N-1$, and thus applies without modification.
\end{proof}
\begin{remark}
    Recursive feasibility and constraint satisfaction are handled entirely by the MPC layer, while the RL component influences only performance. Consequently, the hybrid scheme inherits recursive feasibility directly from the underlying MPC formulation.
\end{remark}

Beyond recursive feasibility, closed-loop performance remains an important consideration in general cost settings. However, deriving formal performance bounds is difficult due to the complexity of the hybrid action space and the learning-based elements that depend on learning quality. We therefore assess closed-loop performance of the hybrid RL-MPC method empirically in the subsequent section.

\section{NUMERICAL EXAMPLE}
\label{sec:casestudy}
To validate the proposed hybrid RL-MPC framework, we apply it to an F1 race strategy problem, with the goal to determine the discrete pit stop decisions and the continuous energy allocation over a race. The problem is a natural fit, involving tightly coupled discrete and continuous actions and a well-defined performance objective. Furthermore, Assumption \ref{ass:one} is satisfied, since for any pit stop decision, the continuous energy allocation can always be chosen to maintain feasibility, as the pit stop decisions affect only the performance objective but not the feasibility of the energy constraints. While race strategies are typically precomputed offline, unforeseen events during the race require fast online re-optimization. This motivates evaluating computation time alongside suboptimality throughout this section.

We adopt the problem setup of \cite{F1RaceStrategy} as our benchmark, where both a MINLP and a standalone RL agent were developed and evaluated.

\subsection{Problem setup}
Since 2014, F1 cars have featured hybrid-electric power units, combining a turbocharged internal combustion engine with an electric motor generator unit connected to a battery. The resulting energy management is subject to strict regulatory bounds on battery capacity, fuel load, and energy deployment \cite{FIAsporting, FIAtechnical}. The race is discretized on a lap-by-lap basis, indexed by $j \in \{0,1,\ldots,N_\mathrm{laps}\}$. The continuous inputs are the allocated battery and fuel energy per lap, $\Delta E_\mathrm{b,all} \in~\left[\Delta E_\mathrm{b,min}, \Delta E_\mathrm{b,max} \right]$ and $\Delta E_\mathrm{f,all} \in~\left[\Delta E_\mathrm{f,min}, \Delta E_\mathrm{f,max} \right]$, while the discrete input is the pit stop decision $\mathrm{PS}\in \left\{0,1,2,3\right\}$, where $\mathrm{PS}=0$ denotes no stop and $\mathrm{PS}\in \{1,2,3\}$ denotes a tire change to the soft~(S), medium (M), or hard (H) compound, respectively. The physical states are the battery energy $E_\mathrm{b}$, fuel energy $E_\mathrm{f}$, car mass $m_\mathrm{car}$, tire wear $\mathrm{TW}$, tire compound~$\mathrm{TC}$, and race time $T_\mathrm{race}$. Regulatory constraints impose $E_\mathrm{b} \in~\left[0, E_\mathrm{b,max}\right]$ and $E_\mathrm{f} \in \left[0, E_\mathrm{f,race}\right]$, where $E_\mathrm{b,max}$ is the maximum allowed battery capacity and $E_\mathrm{f,race}$ is the maximum allowed fuel load per race. The remaining states are not subject to explicit constraints beyond the system dynamics.

The core of the model is the lap time, which is defined as
\begin{equation}
    \begin{aligned}
        T_\mathrm{lap} = \: & T_\mathrm{nom}\left(E_\mathrm{b}, \Delta E_\mathrm{b,all}, \Delta E_\mathrm{f,all}, m_\mathrm{car}, \mathrm{PS}\right) \\
                          &+\Delta T^{\mathrm{TC}}\left(\mathrm{TW}\right),
    \end{aligned}
    \label{eq:Tlap}
\end{equation}
where $T_\mathrm{nom}$ maps the given input to the nominal lap time. The second term $\Delta T^{\mathrm{TC}}$ relates tire wear $\mathrm{TW}$ and tire compound $\mathrm{TC}$ to additional lap time. Tire wear evolves according to a degradation model that captures compound-specific degradation rates and dependence on vehicle mass~\cite{F1RaceStrategy}. The optimization objective is to minimize the total race time~$T_\mathrm{race}$.

Building on the problem setup of \cite{F1RaceStrategy}, we introduce two modifications. First, we refine the MINLP formulation, improving its numerical conditioning and reducing the solve time from \qty{50}{s} to approximately \qty{3}{s}. Second, we use a narrower and deeper critic architecture than~\cite{F1RaceStrategy}, which improves the performance of the trained RL agent. 

For the hybrid RL-MPC framework, we use the Soft Actor-Critic (SAC) algorithm, which offers strong sample efficiency through off-policy training and was introduced for hybrid action spaces in \cite{HybridSAC}, building on \cite{SAC} and \cite{DiscreteSAC}. The critic is trained with ReLU activation functions. Where it is used as a terminal cost, these are replaced at deployment by a smooth approximation~\cite{relusmoothing} to ensure continuous differentiability required by gradient-based NLP solvers, which we verified leaves the cost-to-go estimate essentially unchanged.

\textit{Implementation details:} The RL agent is trained following the setup of \cite{F1RaceStrategy}, with the modified critic architecture described above, using a discount factor of $\gamma=0.9999$. Since the race is finite, the receding-horizon scheme transitions to a shrinking horizon in the final $N$ laps, where no terminal cost is applied as no cost is incurred beyond the end of the race. The reduced NLP \eqref{eq:reduced_MPC} is implemented in CasADi \cite{casadi} and solved to convergence with IPOPT \cite{ipopt} on a commercial laptop (Lenovo ThinkPad, Intel Core Ultra 7, 64 GB RAM), yielding a local rather than a guaranteed global optimum as with any nonconvex NLP.

\subsection{Numerical results}
\label{subsec:benchmark}
\subsubsection{Baselines} We benchmark the hybrid RL-MPC framework against the MINLP optimal solution and the standalone RL agent. The MINLP is solved offline and serves as the optimal baseline that online methods cannot be expected to match in general. The standalone RL method is evaluated online in closed loop. Suboptimalities and mean computation times are summarized in Table~\ref{tab:trace_tcomp_nom}. For the RL agent, the computation time corresponds to a single forward pass, with the full rollout over $N$ steps scaling approximately linearly. For the hybrid RL-MPC, it is the NLP solve time, and for the MINLP, its solve time.

Selected race strategies are shown in Fig.~\ref{fig:results_benchmark}, where the first plot shows the fuel allocation, which is expressed as a percentage of the nominal allocation, where $100\%$ corresponds to the constant per-lap allocation that exactly depletes the fuel budget over the race. The second plot shows the battery energy allocation, where positive values denote charging and negative values discharging. The pit stop plot at the bottom shows the timing and compound selection.

\begin{table}
    \vspace{0.2cm}
    \caption{Performance across methods, horizons, and ablations.}
    \label{tab:trace_tcomp_nom}
    \vspace{-0.1cm}
    \centering
    \begin{threeparttable}
    \begin{tabular}{lcc}
        \toprule
        Method & $\Delta T_\mathrm{race}$ & ${\Bar{T}_\mathrm{comp}}$ \\
        \midrule
        MINLP (offline) & +\qty{0.000}{s} & \qty{3.370}{s} \\ 
        RL & +\qty{1.170}{s} & \qty{0.021}{s} \\
        \midrule
        \textit{Hybrid RL-MPC}\tnote{1} & & \\
        $N=5$ & +\qty{0.784}{s}&  \qty{0.013}{s}  \\
        $N=10$ &+\qty{0.320}{s} &  \qty{0.020}{s} \\
        $N=15$ & +\qty{0.095}{s} &  \qty{0.026}{s} \\
        $N=40$ & +\qty{0.035}{s} &  \qty{0.069}{s} \\
        \midrule
        \textit{Ablation} & & \\
        w/o actor warm-start $N=5$ & +\qty{1.478}{s} & \qty{0.013}{s}\\
        w/o actor warm-start $N=15$ & +\qty{0.095}{s} & \qty{0.029}{s}\\
        w/ critic $N=5$ & +\qty{0.849}{s} & \qty{0.114}{s}\\
        w/ critic $N=15$ & +\qty{0.094}{s} & \qty{0.246}{s} \\
        \bottomrule
    \end{tabular}
    \begin{tablenotes}
        \footnotesize
        \item[1] With an actor warm-start and zero terminal cost, as identified in the ablation study.
    \end{tablenotes}
    \end{threeparttable}
    \vspace{-0.25cm}
\end{table}

The standalone RL agent achieves a suboptimality of~$+$\qty{1.17}{s} relative to the MINLP. To put this in context, F1 races are typically decided by margins of a few seconds, making this gap strategically relevant. The source of suboptimality lies in the continuous energy allocation, as the pit stop decisions are identical to the MINLP solution. This confirms that the discrete decisions are well-learned and the continuous actions require further refinement. As visible in Fig.~\ref{fig:results_benchmark}, the fuel allocation follows a bang-bang strategy rather than a gradual decrease as in the MINLP solution, and the battery is charged during the lap after the pit stop rather than during the pit stop lap, missing the recuperation potential of braking into the pit lane.

\begin{figure}
    \vspace{0.15cm}
    \centering
    \hspace{-1cm}
        \begin{tikzpicture}[trim axis right]

\ifdefined\isMainDocument
  \def\datapath{Results_Benchmark/}
  \def\width{\columnwidth}
  \def\height{0.35*\columnwidth}
\else
  \def\datapath{}
  \def\width{15cm}
  \def\height{5cm}
\fi

\begin{groupplot}[
    group style={
        group size=1 by 3,
        vertical sep=0.5cm,
        xlabels at=edge bottom,
        xticklabels at=edge bottom,
    },
    width=\width,
    height=\height,
    xmin=0, xmax=56,
    xtick={0,10,20,30,40,50},
    xmajorgrids,
    ymajorgrids,
    tick label style={font=\scriptsize},
    tick style={semithick, color=black},
    axis background/.style={fill=white},
]

\nextgroupplot[
    ylabel=\small{$\Delta E_\mathrm{f,all}$},
    y unit=-,
    ytick={90,100,110},
    yticklabels={90\%,100\%,110\%},
    ymin=85, ymax=115,
    ylabel style={yshift=-30pt},
    yticklabel style={text width=1.5cm, align=right},
    legend style={at={(1,1.6)}, anchor=north east, legend cell align=left,
                  align=left, draw=white!15!black, font=\scriptsize,
                  /tikz/every even column/.append style={column sep=2mm}},
]
\addplot [color=gray, dashed, line width=0.5pt, forget plot]
    table[]{lower_bound_fuel.tsv};
\addplot [color=gray, dashed, line width=0.5pt, forget plot]
    table[]{upper_bound_fuel.tsv};
\addplot [color=black, line width=1.5pt]
    table[]{Results_Benchmark/hybrid_vs_minlp_1.tsv};
\addlegendentry{MINLP}
\addplot [color=red, line width=0.75pt]
    table[]{Results_Benchmark/RL_nominal_actions-1.tsv};
\addlegendentry{RL}
\addplot [color=blue, line width=0.75pt]
    table[]{Results_Benchmark/hybrid_vs_minlp-2.tsv};
\addlegendentry{Hybrid RL-MPC $N_{15}$}
\addplot [color=cyan, dashed, line width=0.75pt]
    table[]{Results_Benchmark/hybrid_vs_minlp_nhor40-2.tsv};
\addlegendentry{Hybrid RL-MPC $N_{\geq40}$}

\nextgroupplot[
    ylabel=\small{$\Delta E_\mathrm{b,all}$},
    y unit=-,
    ylabel style={yshift=-30pt},
    ymin=-1.3, ymax=1.3,
    yticklabel style={text width=1.5cm, align=right},
]
\addplot [color=gray, dashed, line width=0.5pt, forget plot]
    table[]{upper_bound_bat.tsv};
\addplot [color=gray, dashed, line width=0.5pt, forget plot]
    table[]{lower_bound_bat.tsv};
\addplot [color=black, line width=1.5pt, forget plot]
    table[]{Results_Benchmark/hybrid_vs_minlp-5.tsv};
\addplot [color=red, line width=0.75pt]
    table[]{Results_Benchmark/RL_nominal_actions-4.tsv};
\addplot [color=blue, line width=0.75pt, forget plot]
    table[]{Results_Benchmark/hybrid_vs_minlp-6.tsv};
\addplot [color=cyan, dashed, line width=0.75pt, forget plot]
    table[]{Results_Benchmark/hybrid_vs_minlp_nhor40-6.tsv};

\nextgroupplot[
    xlabel={Lap},
    ymin=-0.3, ymax=3.3,
    ytick={0,1,2,3},
    yticklabels={No pit, Pit for $S$, Pit for $M$, Pit for $H$},
    yticklabel style={text width=1.5cm, align=right},
    xticklabels={0,10,20,30,40,50},
]
\addplot [ycomb, color=black, mark=o, mark size=1.5pt, line width=0.5pt]
    table[]{Results_Benchmark/hybrid_vs_minlp-9.tsv};
\addplot [ycomb, color=red, mark=o, mark size=1.5pt, line width=0.5pt]
    table[]{Results_Benchmark/RL_nominal_actions-7.tsv};
\addplot [ycomb, color=blue, mark=o, mark size=2.5pt, line width=0.5pt]
    table[]{Results_Benchmark/hybrid_vs_minlp-10.tsv};
\addplot [ycomb, color=cyan, mark=o, mark size=2.5pt, line width=0.5pt]
    table[]{Results_Benchmark/hybrid_vs_minlp_nhor40-10.tsv};

\end{groupplot}
\end{tikzpicture}
        \vspace{-0.3cm}
    \caption{Comparison of the MINLP offline benchmark (black), the trained RL agent (red), the online hybrid RL-MPC method for horizons of $N=15$ laps (blue) and $N\geq 40$ laps (cyan, dashed), with constraint bounds shown as gray dashed lines.}
    \label{fig:results_benchmark}
    \vspace{-0.5cm}
\end{figure}

\subsubsection{Ablation study} The framework of Section~\ref{sec:framework} supplies the MPC with a discrete action sequence, a continuous warm-start, and optionally the critic as terminal cost. How much the latter two contribute is problem-dependent~\cite{AC4MPC}, and we quantify both. Since a shorter horizon leaves the MPC less room to correct the learned components, we evaluate them at a short and a medium horizon, $N=5$ and $N=15$.

According to Table~\ref{tab:trace_tcomp_nom}, the warm-start is decisive at short horizons, where at $N=5$ the cold-started variant has a race time that is \SI{694}{\milli\second} slower. Local minima arise primarily from the discrete pit stop decisions, which are supplied by the RL actor, so at longer horizons the warm-start no longer improves the solution and reduces to a modest computational benefit of a few milliseconds.

Table~\ref{tab:trace_tcomp_nom} shows that the critic as terminal cost provides no measurable performance benefit at horizons of $N=15$ and above, where the two variants differ by about \SI{1}{\milli\second}, and at $N=5$ it degrades the solution with a race time slower by \SI{65}{\milli\second}. At all horizons, the mean solve time increases by roughly an order of magnitude, since the critic network must be evaluated within the optimization. This is consistent with~\cite{AC4MPC}. When the actor combined with MPC refinement is close to optimal, value augmentation offers limited benefit while introducing the critic's approximation error, an effect amplified as the horizon shortens.

That the critic adds little at deployment reflects the quality of the actor produced by the actor-critic training. Because the critic approximates the state-action cost function of the underlying hybrid MDP, the actor is guided during training toward optimizing that cost function directly. This is what enables the actor to reach the observed level of performance at deployment. The critic also provides a terminal cost whenever the actor is not close to optimal. Replacing it with hand-designed terminal costs based on the remaining number of laps and the deviation from a nominal energy allocation loses approximately \SI{5}{\second}. An accurate terminal cost for this economic objective would have to account for effects such as tire degradation and pit stop timing, which the learned critic captures without explicit design.

\subsubsection{Hybrid RL-MPC performance} Based on the ablation study, we analyze the framework with the actor warm-start and a zero terminal cost. At $N\geq40$, above which suboptimality and strategy remain constant, the framework achieves a residual gap of only $+$\qty{35}{ms}, which is negligible given a total race time of approximately 90 minutes. The remaining differences between the MINLP and the long-horizon hybrid method visible in Fig.~\ref{fig:results_benchmark} are due to numerical discrepancies between the offline MINLP solution and the closed-loop receding-horizon implementation. 

Table \ref{tab:trace_tcomp_nom} reveals a clear trade-off between suboptimality and computation time. As the horizon increases from $N=5$ to $N=40$, the mean NLP solve time grows from \qty{13}{ms} to \qty{69}{ms}, while suboptimality decreases from $+$\qty{0.784}{s} to $+$\qty{0.035}{s}, reflecting the truncation of the planning horizon, which is consistent with theoretical predictions of~\cite{BertsekasAlphaZero, AC4MPC}. In the following, we proceed with $N=15$, which achieves a suboptimality of $+$\qty{95}{ms}, a $92\%$ reduction in suboptimality over the standalone RL agent, at a mean solve time of~\qty{26}{ms}. As Fig.~\ref{fig:results_benchmark} shows, it recovers the gradual fuel decrease and the battery recuperation during the pit stop lap that the standalone RL agent misses. The residual gap is mainly due to suboptimal fuel allocation before the first pit stop, forcing a lower allocation in the remainder of the race compared to the MINLP.

Once trained, the RL actor yields the discrete decisions through a forward pass, leaving an integer-free NLP whose size grows linearly in $N$. A receding-horizon mixed-integer formulation would instead require a search over $|U_\mathrm{d}|^N$~discrete sequences at each step. The hybrid RL-MPC framework thus bridges the gap between the fast but suboptimal standalone RL agent and complex, time-consuming mixed-integer optimization.

\subsection{Architectural adaptability}

We demonstrate the adaptability of the hybrid RL-MPC framework by extending it to a traffic scenario without retraining the RL agent. We compare a nominal variant, the implementation from Section~\ref{subsec:benchmark} with $N=15$ operated without any modification, against a modified variant in the same closed-loop environment with traffic present. For the modified variant, we incorporate a model of the traffic into the MPC formulation, demonstrating that the framework adapts to new scenarios at zero retraining cost and low implementation effort.

Traffic arises when a pit stop takes longer than expected and the driver rejoins behind slower competitors, increasing their lap time through the aerodynamic wake of the car ahead. Critically, traffic is a situation where an additional pit stop is undesirable, since the time lost in the pit lane would far outweigh the benefit of rejoining on clear track, meaning that the only available lever is the continuous energy allocation.

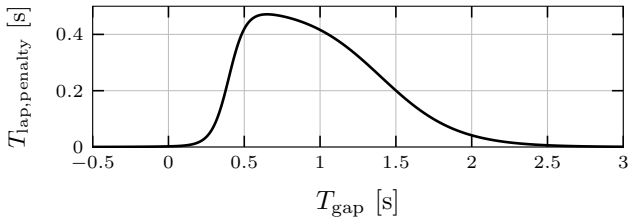
\begin{figure}
    \centering
    \vspace{0.05cm}
    \hspace{-1.0cm}
    \begin{tikzpicture}[trim axis right]

\ifdefined\isMainDocument
  \def\datapath{LapTimePenalty/}
  \def\width{\columnwidth}
  \def\height{0.4*\columnwidth}
\else
  \def\datapath{}
  \def\width{15cm}
  \def\height{5cm}
\fi

\begin{groupplot}[
    group style={
        group size=1 by 1,
        vertical sep=0.5cm,
        xlabels at=edge bottom,
        xticklabels at=edge bottom,
    },
    width=\width,
    height=\height,
    xmin=-0.5, xmax=3.0,
    xmajorgrids,
    ymajorgrids,
    tick label style={font=\scriptsize},
    tick style={semithick, color=black},
    axis background/.style={fill=white},
]

\nextgroupplot[
    ylabel=\small{$T_\mathrm{lap,penalty}$},
    y unit=s,
    xlabel={$T_\mathrm{gap}$},
    x unit=s,
    ymin=0, ymax=0.5,
    ylabel style={yshift=-30pt},
    yticklabel style={text width=1.5cm, align=right},
]
\addplot [color=black, line width=1.0pt]
    table[]{LapTimePenalty/traffic_penalty-1.tsv};

\end{groupplot}
\end{tikzpicture}
    \vspace{-0.3cm}
    \caption{Lap time penalty as a function of gap time to the car ahead, with a maximum of approximately \qty{0.5}{s} at a gap of \qty{0.6}{s}, vanishing both in close proximity to the competitor and beyond \qty{2.5}{s}.}
    \label{fig:LapTimePenalty}
    \vspace{-0.5cm}
\end{figure}
In the following, we introduce traffic after the first pit stop. The lap time penalty $T_\mathrm{lap,penalty}$ due to traffic, shown in Fig.~\ref{fig:LapTimePenalty}, is modeled as a heuristic function of the gap time to the car ahead, peaking at approximately \qty{0.5}{s} for a gap of \qty{0.6}{s}, vanishing beyond \qty{2.5}{s} as the wake dissipates, and dropping to zero at small gaps as the cars are side by side. The gap time is introduced as a new continuous state in the MPC layer and updated via
\begin{equation}
    T_\mathrm{gap}(j+1) = T_\mathrm{gap}(j) + T_\mathrm{lap,traffic}(j) - T_\mathrm{lap,opp}(j),
    \label{eq:Tgap}
\end{equation}
initialized at 2 seconds at the lap of the first pit stop, where $T_\mathrm{lap,opp}$ is the opponent's lap time, which we model as \qty{0.5}{s} slower than the lap time in a nominal scenario. The lap time in the traffic scenario then becomes
\begin{equation}
    T_\mathrm{lap,traffic}(j) = T_\mathrm{lap}(j) + T_\mathrm{lap,penalty}\left(T_\mathrm{gap}(j)\right),
\end{equation}
where $T_\mathrm{lap}$ follows from \eqref{eq:Tlap}. The nominal implementation has no model of \eqref{eq:Tgap} in its MPC formulation and therefore no estimate of the lap time penalty $T_\mathrm{lap,penalty}$, while the modified variant incorporates \eqref{eq:Tgap} as an additional state. In both cases, the RL agent is identical to the one in the nominal scenario presented in Section \ref{subsec:benchmark}.

Since the modified variant explicitly predicts the traffic penalty, it will naturally outperform the nominal one. The goal of this case study is to quantify how much performance a simple extension recovers at zero retraining cost.

\begin{figure}
    \centering
    \hspace{-1cm}
    \begin{tikzpicture}[trim axis right]

\ifdefined\isMainDocument
  \def\datapath{Results_traffic/}
  \def\width{\columnwidth}
  \def\height{0.35*\columnwidth}
\else
  \def\datapath{}
  \def\width{15cm}
  \def\height{5cm}
\fi

\begin{groupplot}[
    group style={
        group size=1 by 4,
        vertical sep=0.5cm,
        xlabels at=edge bottom,
        xticklabels at=edge bottom,
    },
    width=\width,
    height=\height,
    xmin=0, xmax=56,
    xtick={0,10,20,30,40,50},
    xmajorgrids,
    ymajorgrids,
    tick label style={font=\scriptsize},
    tick style={semithick, color=black},
    axis background/.style={fill=white},
]

\nextgroupplot[
    ylabel=\small{$\Delta E_\mathrm{f,all}$},
    y unit=-,
    ytick={90,100,110},
    yticklabels={90\%,100\%,110\%},
    ymin=85, ymax=115,
    ylabel style={yshift=-30pt},
    yticklabel style={text width=1.5cm, align=right},
    legend style={at={(1,1)}, anchor=north east, legend cell align=left,
                  align=left, draw=white!15!black, font=\scriptsize,
                  /tikz/every even column/.append style={column sep=2mm}},
]
\addplot [color=gray, dashed, line width=0.5pt, forget plot]
    table[]{lower_bound_fuel.tsv};
\addplot [color=gray, dashed, line width=0.5pt, forget plot]
    table[]{upper_bound_fuel.tsv};
\addplot [color=black, line width=1.0pt]
    table[]{Results_traffic/traffic_hybrid_actions-1.tsv};
\addlegendentry{No traffic model}
\addplot [color=blue, dashed, line width=1.0pt]
    table[]{Results_traffic/traffic_hybrid_actions-2.tsv};
\addlegendentry{With traffic model}

\nextgroupplot[
    ylabel=\small{$\Delta E_\mathrm{b,all}$},
    y unit=-,
    ylabel style={yshift=-30pt},
    ymin=-1.3, ymax=1.3,
    yticklabel style={text width=1.5cm, align=right},
]
\addplot [color=gray, dashed, line width=0.5pt, forget plot]
    table[]{upper_bound_bat.tsv};
\addplot [color=gray, dashed, line width=0.5pt, forget plot]
    table[]{lower_bound_bat.tsv};
\addplot [color=black, line width=1.0pt, forget plot]
    table[]{Results_traffic/traffic_hybrid_actions-5.tsv};
\addplot [color=blue, dashed, line width=1.0pt, forget plot]
    table[]{Results_traffic/traffic_hybrid_actions-6.tsv};

\nextgroupplot[
    ymin=-0.3, ymax=3.3,
    ytick={0,1,2,3},
    yticklabels={No pit, Pit for $S$, Pit for $M$, Pit for $H$},
    yticklabel style={text width=1.5cm, align=right},
]
\addplot [ycomb, color=black, mark=o, mark size=1.5pt, line width=0.5pt]
    table[]{Results_traffic/traffic_hybrid_actions-9.tsv};
\addplot [ycomb, color=blue, mark=o, mark size=2.5pt, line width=0.5pt]
    table[]{Results_traffic/traffic_hybrid_actions-10.tsv};

\nextgroupplot[
    xlabel={Lap},
    ylabel=\small{$T_\mathrm{gap}$},
    y unit= s,
    ylabel style={yshift=-30pt},
    ymin=-8, ymax=2.5,
    yticklabel style={text width=1.5cm, align=right},
    xticklabels={0,10,20,30,40,50},
]
\addplot [color=black, line width=1.0pt, forget plot]
    table[]{Results_traffic/traffic_hybrid_gap-1.tsv};
\addplot [color=blue, dashed, line width=1.0pt, forget plot]
    table[]{Results_traffic/traffic_hybrid_gap-2.tsv};

\end{groupplot}
\end{tikzpicture}
    \vspace{-0.3cm}
    \caption{Comparison of the hybrid RL-MPC without traffic model (black) and with traffic model \eqref{eq:Tgap} (blue), with constraint bounds shown as gray dashed lines. The bottom plot shows the gap time to the competitor, where positive values indicate being behind.}
    \vspace{-0.5cm}
    \label{fig:results_traffic}
\end{figure}

Fig. \ref{fig:results_traffic} compares the resulting race strategies, showing the same quantities as in Fig. \ref{fig:results_benchmark} and the gap time to the competitor. Both implementations keep the pit stop strategy identical to the traffic-free scenario from Fig.~\ref{fig:results_benchmark}, confirming that an additional stop is not beneficial given the time lost in the pit lane.

The nominal hybrid RL-MPC retains the energy allocation of the nominal scenario in Fig. \ref{fig:results_benchmark}, as it does not have a traffic model. The car eventually overtakes by virtue of the competitor being slower. However, this happens only much later, accumulating a substantial lap time penalty of \qty{6.87}{s} compared to the traffic-free scenario.

The modified hybrid RL-MPC responds immediately through the MPC layer. Right after the first pit stop, the modified controller increases fuel consumption and recharges the battery, exploiting load-point shifting to harvest electric energy while incurring the traffic penalty anyway. Once sufficient battery energy is accumulated, it is discharged aggressively to overtake the traffic and stay ahead, after which the energy allocation returns to the nominal strategy. The overtake occurs significantly earlier, and the accumulated lap time penalty is substantially reduced, with the modified implementation finishing \qty{5.19}{s} ahead of the nominal one, which is a race-deciding margin. 

This result demonstrates that the MPC layer can be adapted to new scenarios by incorporating scenario-specific constraints without modifying or retraining the RL agent.

\section{CONCLUSION}
\label{sec:conclusion}
We proposed a hybrid RL-MPC framework for mixed-integer optimal control. By training the RL agent on the full hybrid action space, the learned critic approximates the state-action cost function of the underlying hybrid MDP. During deployment, the RL actor is rolled out over the prediction horizon to parametrize a continuous NLP-MPC through the discrete action sequence. Additionally, it provides a warm-start for the continuous variables, and the learned critic can serve as a terminal cost. We proved recursive feasibility under a structural assumption on the discrete decisions, and validated the framework on an F1 race strategy problem, outperforming a standalone RL agent and achieving near-optimal performance relative to an offline MINLP benchmark. An ablation study showed that for the F1 race setting the actor warm-start is decisive at short horizons, while the critic as terminal cost provided no benefit at the horizons considered, though its role in training the actor remains essential. The hybrid architecture further enables adaptation to unseen disturbances, demonstrated on a traffic scenario, yielding a \qty{5.19}{s} advantage without any retraining of the RL components.

Two directions are of particular interest for future work. First, deriving formal suboptimality bounds similar to \cite{AC4MPC} would strengthen the theoretical foundation, and a quantitative comparison against a receding-horizon mixed-integer formulation would complement the offline benchmark used here. Second, relaxing Assumption \ref{ass:one} to settings where discrete decisions affect feasibility could be addressed through safety filters \cite{safetyfilter} or control barrier functions \cite{cbf}, extending the framework to broader hybrid action space systems.

\section*{ACKNOWLEDGMENTS}
We thank Ilse New for her careful proofreading and helpful comments, as well as Marc-Philippe Neumann and Mohammad Moradi for fruitful discussions and feedback.

\textit{Use of generative artificial intelligence:} The chatbot Claude~(Anthropic) was used to assist with grammar and phrasing in selected paragraphs of this manuscript.

\end{document}